\documentclass[letterpaper, 11pt]{llncs}

\usepackage[margin=1in]{geometry}
\usepackage[bookmarks, colorlinks=true, plainpages = false, citecolor = blue,linkcolor=red,urlcolor = blue, filecolor = blue]{hyperref}
\usepackage{url}

\usepackage{amsmath,amsfonts,amssymb,bm, verbatim,mathtools}

\usepackage{color,graphicx,appendix}
\usepackage[linesnumbered,ruled]{algorithm2e}

\usepackage{subfigure}
\usepackage{etoolbox}
\usepackage{tikz}
\usepackage{xr,xspace}
\usepackage{todonotes}
\usepackage{paralist}
\usepackage{caption}

\newcommand{\mycomment}[1]{}

\usepackage{xspace,prettyref}
\newcommand{\diverge}{\to\infty}

\newcommand{\reals}{{\mathbb{R}}}

\newcommand{\toas}{\xrightarrow{{\rm a.s.}}}

\newrefformat{eq}{(\ref{#1})}
\newrefformat{chap}{Chapter~\ref{#1}}
\newrefformat{sec}{Section~\ref{#1}}
\newrefformat{algo}{Algorithm~\ref{#1}}
\newrefformat{fig}{Fig.~\ref{#1}}
\newrefformat{tab}{Table~\ref{#1}}
\newrefformat{rmk}{Remark~\ref{#1}}
\newrefformat{clm}{Claim~\ref{#1}}
\newrefformat{def}{Definition~\ref{#1}}
\newrefformat{cor}{Corollary~\ref{#1}}
\newrefformat{lmm}{Lemma~\ref{#1}}
\newrefformat{prop}{Proposition~\ref{#1}}
\newrefformat{app}{Appendix~\ref{#1}}
\newrefformat{hyp}{Hypothesis~\ref{#1}}
\newrefformat{thm}{Theorem~\ref{#1}}

\newcommand{\pth}[1]{\left( #1 \right)}

\newcommand{\sth}[1]{\left\{ #1 \right\}}

\newcommand{\norm}[1]{\left\|{#1} \right\|}

\newcommand{\iprod}[2]{\left \langle #1, #2 \right\rangle}

\newcommand{\sfB}{{\mathsf{B}}}

\newcommand{\calE}{{\mathcal{E}}}

\newcommand{\calI}{{\mathcal{I}}}

\newcommand{\calO}{{\mathcal{O}}}

\newcommand{\calV}{{\mathcal{V}}}

\newcommand{\calX}{{\mathcal{X}}}

\newcommand{\argmin}{{\rm argmin}}

\begin{document}

\title{On the Convergence Rate of Average Consensus and Distributed Optimization over Unreliable Networks\thanks{This research is supported in part by National Science Foundation awards NSF 1329681 and 1421918.
Any opinions, findings, and conclusions or recommendations expressed here are those of the authors
and do not necessarily reflect the views of the funding agencies or the U.S. government.} 
\thanks{email: lilisu3@csail.mit.edu}}
\date{\today}

\author{Lili Su\\
\institute{EECS, MIT}
}
\maketitle

\begin{abstract}
We consider the problems of reaching average consensus and solving consensus-based optimization over unreliable communication networks wherein packets may be dropped accidentally during transmission.
Existing work either assumes that the link failures affect the communication on both directions or that the message senders {\em know exactly}, in each iteration, how many of their outgoing links are functioning properly. 
In this paper, we consider directed links, and we {\em do not} require each node know its current outgoing degree. 

First, we propose and characterize the convergence rate of reaching average consensus. 
%
Then we apply our robust consensus update to the classical distributed dual averaging method wherein the consensus update is used as the information aggregation primitive. 
%
We show that the local iterates converge to a common optimum of the global objective at rate $O(\frac{1}{\sqrt{t}})$, where $t$ is the number of iterations, matching the failure-free performance of the distributed dual averaging method.

\end{abstract}

\section{Introduction}
\label{sec: intro}
Reaching consensus and solving distributed optimization are two closely related global tasks of multi-agent networks. In the former, every agent has a private input, and the common goal of the networked agents is to reach an agreement on a value that is determined by these private inputs such as maximum, minimum, average, etc.  In the latter, typically, every agent has a private cost function, and the common goal is to collaboratively minimize a global objective which is some proper aggregation of these private cost functions.

Average consensus has received intensive attention \cite{6571230,hadjicostis2016robust,xiao2007distributed} partially because one can use average consensus as a way to aggregate agents' private information. Different strategies to robustify reaching average consensus against unreliable networks have been proposed \cite{patterson2007distributed,fagnani2009average,YinChen2010,vaidya2012robust,Eyal2012}. 
Undirected graphs were considered in \cite{patterson2007distributed,Eyal2012}, where the link failures affect the communication on both directions.  Dynamically changing data and networks are considered in \cite{Eyal2012}.
Directed graphs were first considered in \cite{fagnani2009average}, however, only biased averages were achieved. 
Later work corrected these biases \cite{YinChen2010,vaidya2012robust} via introducing auxiliary variables at each agent; however, only asymptotic convergence was shown. To the best of our knowledge, non-asymptotic convergence rate is still lacking \footnote{As indicated by the common arXiv code, \cite{su2016robustOV} is only an early version of this work.}.

Consensus-based multi-agent optimization is an important family of distributed optimization algorithms. In a typical consensus-based multi-agent optimization problem \cite{Duchi2012,Nedic2009,nedic2015distributed,Tsianos2012CDC}, each agent $i$ keeps a {\em private} cost function $h_i: \calX \to \reals$, 
and the networked agents collectively want to reach agreement on a global decision $x^*\in \calX$ such that the average of these private cost functions is minimized, i.e., 
\begin{align*}
x^* ~ \in ~ \argmin_{x\in \calX} ~ \frac{1}{n} \sum_{i=1}^n h_i(x), 
\end{align*}
where $n$ is the total number of agents in the system. The applications of such distributed optimization problems include distributed machine learning and distributed resource allocation. 
Robustifying distributed optimization against link failures has received some attention recently \cite{Duchi2012,nedic2015distributed}. Duchi et al.\ \cite{Duchi2012} assumed that each realizable link failure pattern admits a doubly-stochastic matrix. 
Assuming each agent knows, in each iteration, the number of outgoing links that are working properly \cite{nedic2015distributed}, the requirement for doubly stochastic matrices was removed by incorporating the push-sum mechanism. However, the implementation of push-sum in \cite{nedic2015distributed} implicitly assumed the adoption of acknowledgement mechanism.

In this work, we consider directed links, and we {\em do not} require each node know its current outgoing degree. 
That is, the message losses might be asymmetric between a pair of agents, and if a message packet is dropped by a link, the sender is not aware of this loss. 
Although acknowledge mechanisms can be incorporated to improve reliability, it may slow down the convergence due to the necessity of message retransmission. 
In this paper, we first propose and characterize the convergence rate of reaching average consensus in the presence of packet-dropping link failures. 
Then we apply our robust consensus update to the classical distributed dual averaging method wherein the consensus update is used as the information aggregation primitive. 
We show that the local iterates converge to a common optimum of the global objective at rate $O(\frac{1}{\sqrt{t}})$, where $t$ is the number of iterations, matching the failure-free performance of the distributed dual averaging method. 

\section{Network Model and Notation}

We consider a synchronous system which consists of $n$ networked agents. The network structure is represented by a {\em strongly connected} graph $G(\calV,\calE)$, where $\calV=\{1,\dots,n\}$ is the collection of agents, and $\calE$ is the collection of {\em directed} communication links. 
Let $\calI_i=\{ j ~|~ (j,i) \in \calE\}$ and $\calO_i=\{j~|~(i,j)\in\calE\}$ be the sets of incoming neighbors and outgoing neighbors, respectively, of agent $i$.
For ease of exposition, we assume no self-loops exist, i.e., $i\notin \calI_i \cup \calO_i,  \forall i\in\calV$.
For $i\in \calV$, let $d^{o}_i=|\calO_i|$. The communication links are unreliable in that packets may be dropped during transmission unexpectedly. However, a given link is assumed to be {\em operational} at least once during $B$ consecutive iterations, where $B\ge 1$.
Similar assumption is adopted in \cite{nedic2015distributed,Nedic2009}.

\section{Robust Average Consensus}
Reaching average consensus in directed networks has been intensively studied \cite{1333204,aysal2009broadcast,kempe2003gossip}.
In particular, in Push-Sum \cite{kempe2003gossip,benezit2010weighted}, each networked agent updates two coupled iterates and the ratio of these two iterates approaches the average asymptotically. The correctness of Push-Sum crucially relies on ``mass preservation"  (specified later) of the system. However, when the communication links suffer packet-dropping failures, the desired ``mass preservation" may not hold. 
Robustification method has been introduced to recover the dropped ``mass" \cite{hadjicostis2016robust,hadjicostis2014average}, where auxiliary variables are introduced to {\em record} the total ``mass" sent and delivered, respectively, through a given communication link. However, only asymptotic convergence is provably guaranteed \cite{hadjicostis2016robust,hadjicostis2014average}. In this section, we focus on characterizing the convergence rate of robust average consensus.
To do that, we need to provide an algorithmic fix of the robust Push-Sum proposed in \cite{hadjicostis2016robust}. This simple fix allows us to use the standard matrix product analysis to show convergence. Note that better convergence rates might be obtained by carefully exploring the structures of the communication graphs.

\subsection{Robust Push-Sum}
In this subsection, we briefly review the Push-Sum algorithm \cite{kempe2003gossip,benezit2010weighted} and its robust variant \cite{hadjicostis2016robust}.
\begin{algorithm}
\caption{Push-Sum \cite{kempe2003gossip,benezit2010weighted}}
\label{ps 1}

{\em Initialization}: $z_i[0]=y_i\in \reals^d$, $w_i[0]=1\in \reals$.\\
%

\For{$t\ge 1$}{
%
Broadcast $\frac{z_i[t-1]}{d_i^{o}+1}$ and $\frac{w_i[t-1]}{d_i^{o}+1}$ to all outgoing neighbors\;
$z_i[t] \gets \sum_{j\in \calI_i \cup \{i\}} \frac{z_j[t-1]}{d_j^{o}+1}$, and $ w_i[t] \gets \sum_{j\in \calI_i \cup \{i\}} \frac{w_j[t-1]}{d_j^{o}+1}$.
}
\end{algorithm}
In standard Push-Sum \cite{kempe2003gossip,benezit2010weighted}, described in Algorithm \ref{ps 1}, each agent $i$ runs two iterates:
\begin{itemize}
\item  value sequence $\{z_i[t]\}_{t=0}^{\infty}$, and
\item  weight sequence $\{w_i[t]\}_{t=0}^{\infty}$,
\end{itemize}
where $z_i[0]=y_i\in \reals^d$ is the private input, and $w_i[0]=1\in \reals$ is the initial weight  of agent $i$. The weight sequences $\{w_i[t]\}_{t=0}^{\infty}$ are introduced to relax the need for doubly stochastic matrices. 
In a sense, the weights are used to correct the ``bias" caused by the network structure.
In each iteration of Algorithm \ref{ps 1}, each agent $i$ divides both the local value $z_i$ and local weight $w_i$ by $d_i^o+1$, recalling that $d_i^o$ is the out-degree of agent $i$ in the fixed $G(\calV, \calE)$. Among the $d_i^o+1$ parts of the value fractions $\frac{z_i}{d_i^o+1}$ and the weight fractions $\frac{w_i}{d_i^o+1}$, agent $i$ sends $d_i^o$ parts to its outgoing neighbors and maintains one part itself. Upon receiving the value fractions and the weight fractions from its incoming neighbors, agent $i$ sums them up respectively.
When the communication network is reliable, 
the ratio of the value and the weight converges to the average of the private inputs, i.e.,
\begin{align}
\label{convergence ps}
 \lim_{t\diverge} \frac{z_i[t]}{w_i[t]} =  \frac{1}{n} \sum_{j=1}^{n} y_j, ~~~\forall ~i =1, \cdots, n.
\end{align}
The correctness of Push-Sum algorithm relies crucially on the {\em mass preservation} of the system \cite{Benezit,kempe2003gossip}, which says that the total weights kept by the agents in the system sum up to $n$ at every iteration, i.e.,
\begin{align}
\label{mass preserv}
\sum_{i=1}^n w_i[t] ~ =~ n, ~~~ \forall ~ t.
\end{align}
Unfortunately, \eqref{mass preserv} does not hold in the presence of packet-dropping link failures. Nevertheless, as illustrated in \cite{hadjicostis2016robust} (also described below in Algorithm \ref{rps 0}), if we are able to keep track of the dropped ``mass", we are able to show that the total mass is preserved in some {\em augmented graph}, where virtual agents/nodes are introduced. 
Though it is tempting to view running Algorithm \ref{rps 0} as running standard push-sum on the augmented graph, this might not be true. As can be seen later, the dynamics under Algorithm \ref{rps 0}, in the current form, are unstable. 
In this paper, we provide a simple algorithmic fix of Algorithm \ref{rps 0}. 

\begin{algorithm}
\caption{Robust Push-Sum \cite{hadjicostis2016robust}}
\label{rps 0}
{\em Initialization}: $z_i[0]=y_i\in \reals^d$, $w_i[0]=1\in \reals,$ $\sigma_i[0]={\bf 0}\in \reals^d$, $\tilde{\sigma}_i[0]=0\in \reals$, and $\rho_{ji}[0]={\bf 0}\in \reals^d$, $\tilde{\rho}_{ji}[0]=0\in \reals$ for each incoming link, i.e., $j\in \calI_i$.

\For{$t\ge 1$}
{\vskip 0.5\baselineskip
$\sigma_i[t]  \gets  \sigma_i[t-1] + \frac{z_i[t-1]}{d_i^{o}+1}$, $
\tilde{\sigma}_i[t] \gets \tilde{\sigma}_i[t-1]+\frac{w_i[t-1]}{d_i^{o}+1}$\;

Broadcast $\pth{\sigma_i[t], \tilde{\sigma}_i[t]}$ to outgoing neighbors\;
\For {each incoming link $(j,i)$}
{\eIf{message $\pth{\sigma_j[t], \tilde{\sigma}_j[t]}$ is received}
{$\rho_{ji}[t] \gets \sigma_j[t]$, ~~ $\tilde{\rho}_{ji}[t] \gets \tilde{\sigma}_j[t]$\;}
{ $\rho_{ji}[t] \gets \rho_{ji}[t-1]$, ~~$\tilde{\rho}_{ji}[t] \gets \tilde{\rho}_{ji}[t-1]$\;}
$ z_i[t] \gets \sum_{j\in \calI_i\cup\{i\}} \pth{\rho_{ji}[t] - \rho_{ji}[t-1]}$, and $w_i[t]  \gets \sum_{j\in \calI_i\cup\{i\}} \pth{\tilde{\rho}_{ji}[t] -\tilde{\rho}_{ji}[t-1]}$.
}
\vskip 0.5\baselineskip
}
\end{algorithm}

Similar to the standard Push-Sum, in Algorithm \ref{rps 0}, each agent $i$ wants to share with its outgoing neighbors of its value fraction $\frac{z_i}{d_i^o+1}$ and weight fraction $\frac{w_i}{d_i^o+1}$. If agent $i$ sends these two fractions out directly, 
the total mass will not be preserved. In order to recover the ``mass" dropped by an incoming link, in addition to $z_i[t]$ and $w_i[t]$, each agent $i$ uses variable $\tilde{\sigma}_i[t]$ to record the cumulative weight (up to iteration $t$) sent through each outgoing link, and uses variable $\sigma_i[t]$ for the corresponding quantity of the value sequence. In particular,
\begin{align}\label{sent cumulative}
\nonumber
 \sigma_i[t] & =  \sigma[t-1] + \frac{z_i[t-1]}{d_i^{o} +1}, ~~ \text{and}\\
\tilde{\sigma}_i[t] & =\tilde{\sigma}_i[t-1] + \frac{w_i[t-1]}{d_i^{o} +1},
\end{align}
with $\sigma_i[0]={\bf 0}\in \reals^d$, and $\tilde{\sigma}_i[0]=0\in \reals$. In each iteration, agent $i$ broadcasts the tuple $\pth{\sigma_i[t], \tilde{\sigma}_i[t]}$ to all of its outgoing neighbors. 
To record the cumulative information {\em delivered} via the link $(i,k)$, the outgoing neighbor $k$ uses a pair of variables $\rho_{ik}[t]$ and $\tilde{\rho}_{ik}[t]$, with $\rho_{ik}[0]={\bf 0}\in \reals^d$ and $\tilde{\rho}_{ik}[0]=0\in \reals$. If the link $(i,k)$ is operational, i.e., the tuple $\pth{\sigma_i[t], \tilde{\sigma}_i[t]}$ is successfully delivered, then
\begin{align*}
\rho_{ik}[t]  = \sigma_i[t], ~\text{and} ~   \tilde{\rho}_{ik}[t]  = \tilde{\sigma}_i[t].
\end{align*}
Otherwise, since no new message is delivered, both $\rho_{ik}[t]$ and $\tilde{\rho}_{ik}[t]$ are unchanged. In summary, if the link is operational at a given iteration, then
\begin{align*}
 \text{total ``mass" sent} ~~ = ~~ \text{total ``mass" delivered};
\end{align*}
Otherwise,
\begin{align*}
 \text{total ``mass" sent} ~~ \not= ~~ \text{total ``mass" delivered}.
\end{align*}
%
In addition, if the link $(i,k)$ is operational at iteration $t$, it holds that
\begin{align}
\label{s1}
\rho_{ik}[t] - \rho_{ik}[t-1] ~~ &= ~~ \sum_{r=t^{\prime}}^{t-1} ~\frac{z_i[r]}{d_i^{o}+1}, ~~\text{and}\\
\label{s2}
\tilde{\rho}_{ik}[t] -\tilde{\rho}_{ik}[t-1] ~~ &= ~~  \sum_{r=t^{\prime}}^{t-1} ~\frac{w_i[r]}{d_i^{o}+1},
\end{align}
where $t^{\prime}$ is the immediately preceding iteration of $t$ such that link $(i,k)$ is operational. As a link is reliable at least once during $B$ consecutive iterations, it holds that $t-t^{\prime}\le B$.
Under Algorithm \ref{rps 0}, it has been shown that \cite{hadjicostis2016robust}, at each agent $i$,
\begin{align*}
\frac{z_i[t]}{w_i[t]} ~ \toas ~ \frac{1}{n}\sum_{i=1}^n y_i, ~~ \text{as }t\diverge.
\end{align*}
However, no convergence rate (asymptotic or non-asymptotic) is given. Informally speaking, this is because the dynamics of the system under Algorithm \ref{rps 0} is not stable enough. In particular, in the augmented graph constructed in \cite{hadjicostis2016robust} (formally defined later), the two iterates ``kept" by the virtual agents are reset to zero periodically and unexpectedly. 
This ``reset" causes non-trivial technical challenges. In particular, the corresponding matrix product does not converge to a rank one matrix.

\subsection{Convergent Robust Push-Sum}
In this subsection, we propose a simple algorithmic fix of Algorithm \ref{rps 0}. We refer to our algorithm as {\em Convergent Robust Push-Sum}, described in Algorithm \ref{alg:ps convergence rate}. 
Note that this does not mean that our  Algorithm \ref{alg:ps convergence rate} is superior or inferior to Algorithm 2  \cite{hadjicostis2016robust}. 

Our Algorithm \ref{alg:ps convergence rate} has the same set of variables as that in Algorithm \ref{rps 0}. For ease of exposition, we use $\sigma^{+}_i[t]$, $\tilde{\sigma}^{+}_i[t]$, $z_i^+[t]$, and $w_i^{+}[t]$ to emphasize the fact that they are intermediate values of corresponding quantities in an iteration. 
%
\begin{algorithm}
\caption{Convergent Robust Push-Sum}
\label{alg:ps convergence rate}
{\em Initialization}: $z_i[0]=y_i\in \reals^d$, $w_i[0]=1\in \reals,$ $\sigma_i[0]={\bf 0}\in \reals^d$, $\tilde{\sigma}_i[0]=0\in \reals$, and $\rho_{ji}[0]={\bf 0}\in \reals^d$, $\tilde{\rho}_{ji}[0]=0\in \reals$ for each incoming link, i.e., $j\in \calI_i$.

\For{$t\ge 1$}
{
$\sigma^{+}_i[t]  \gets  \sigma_i[t-1] + \frac{z_i[t-1]}{d_i^{o}+1}$,
$\tilde{\sigma}^{+}_i[t]  \gets  \tilde{\sigma}_i[t-1] + \frac{w_i[t-1]}{d_i^{o}+1}$\;

Broadcast $\pth{\sigma^{+}_i[t], \tilde{\sigma}^{+}_i[t]}$ to outgoing neighbors\;

\For {each incoming link $(j,i)$}
{\eIf{message $\pth{\sigma^{+}_j[t], \tilde{\sigma}^{+}_j[t]}$ is received}
{$\rho_{ji}[t] \gets \sigma^{+}_j[t]$, ~~ $\tilde{\rho}_{ji}[t] \gets \tilde{\sigma}^{+}_j[t]$\;}
{ $\rho_{ji}[t] \gets \rho_{ji}[t-1]$, ~~$\tilde{\rho}_{ji}[t] \gets \tilde{\rho}_{ji}[t-1]$\;}
$ z_i^{+}[t] \gets \frac{z_i[t-1]}{d_i^{o}+1} +  \sum_{j\in \calI_i} \pth{\rho_{ji}[t] - \rho_{ji}[t-1]}$, 
$w_i^{+}[t]  \gets \frac{w_i[t-1]}{d_i^{o}+1} + \sum_{j\in \calI_i} \pth{\tilde{\rho}_{ji}[t] -\tilde{\rho}_{ji}[t-1]}$.
}

$\sigma_i[t]  \gets  \sigma^{+}_i[t] + \frac{z_i^{+}[t]}{d_i^{o}+1}$,
$\tilde{\sigma}_i[t]  \gets  \tilde{\sigma}^{+}_i[t] + \frac{w_i^{+}[t]}{d_i^{o}+1}$,
$z_i[t]  \gets \frac{z_i^+[t]}{d_i^{o}+1}$,
$w_i[t] \gets \frac{w_i^+[t]}{d_i^{o}+1}$.
}
\end{algorithm}
In each iteration of our Algorithm \ref{alg:ps convergence rate}, the cumulative transmitted value and weight $(\sigma_i, \tilde{\sigma}_i)$, and the local value and weight $(z_i, w_i)$ are updated twice, with the first update being identical to that in Algorithm \ref{rps 0}. As mentioned before, with only this first update, the dynamics in the system is not stable enough, as the two iterates ``kept" by the virtual agents are reset to zero periodically and unexpectedly. This ``reset" is prevented by the second update in our Algorithm \ref{alg:ps convergence rate}. Intuitively speaking, in the second update, each agent pushes nonzero ``mass" to the virtual agents on its outgoing links. As a result of this, the two iterates ``kept" by a virtual agent will never be zero at the end of an iteration.
\subsection{Augmented Graph}
%
The augmented graph of a given $G(\calV, \calE)$, denoted as $G^a(\calV^a, \calE^a)$, is constructed as follows \cite{vaidya2012robust}:
\begin{enumerate}
\item $\calV^a=\calV\cup \calE$, i.e., $|\calE|$ additional auxiliary agents are introduced, each of which represents a link in $G(\calV, \calE)$. For ease of notation, we use $n_{ij}$ to denote the virtual agent corresponding to edge $(i,j)$.
\item $\calE\subseteq \calE^a$, i.e., the edge set in $G^a(\calV^a, \calE^a)$ preserves the topology of $G(\calV, \calE)$;
\item Additionally, auxiliary edges are introduced: each auxiliary agent $n_{ij}$ has one incoming neighbor -- agent $i$ -- and one outgoing neighbor -- agent $j$.
\end{enumerate}
\begin{figure}
\centering
\subfigure[Original graph]{
\label{figOR}
\includegraphics[width=1.8 in]{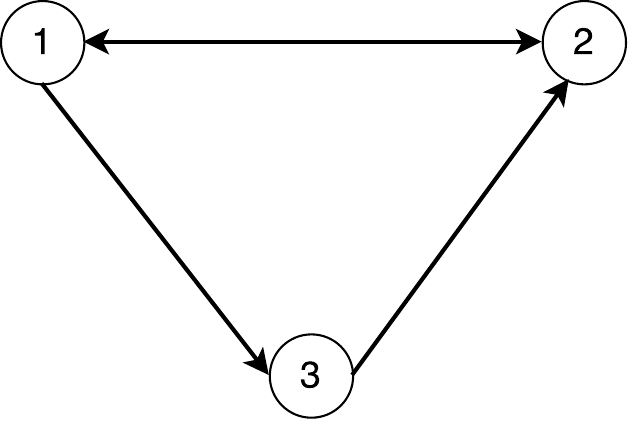}}
\quad \quad \quad \quad 
\subfigure[Augmented graph]{
\label{figORA}
\includegraphics[width=2 in]{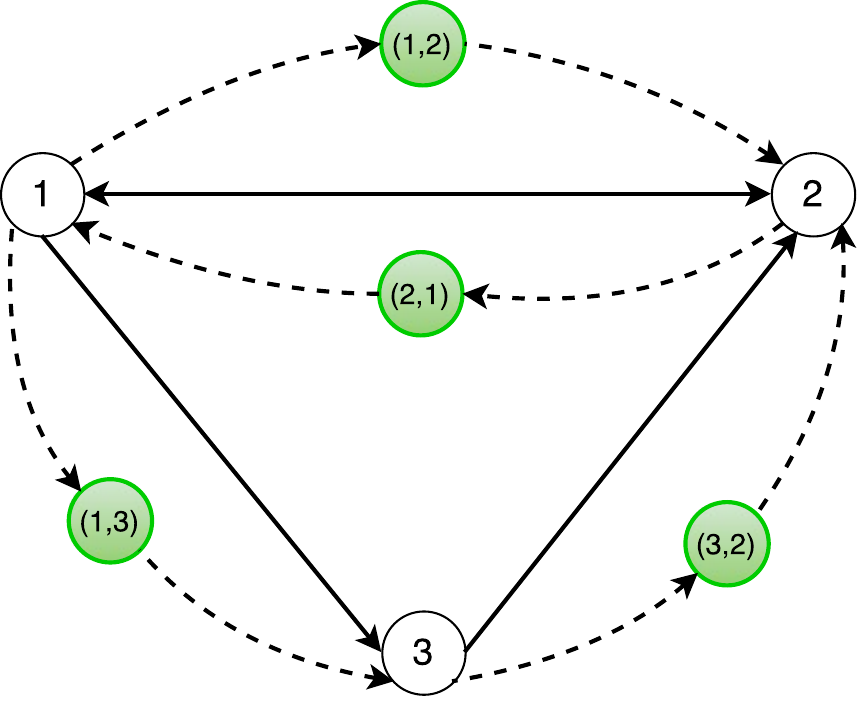}}
\caption{For each directed link, a buffer agent is added.}
\label{Fig.graph}
\end{figure}
As shown in Fig.\ \ref{Fig.graph}, 
in the augmented graph (i.e., Fig.\ \ref{figORA}), four additional agents are introduced, each of which corresponds to a directed edge of the original graph. 
%

\subsection{Matrix Representation}
For each link $(j,i)\in\calE$, and $t\ge 1$, define the indicator variable $\sfB_{(j,i)}[t]$ as follows:
\begin{align}
\sfB_{(j,i)}[t] \triangleq \left\{
  \begin{tabular}{ll}
  1, ~ if link $(j,i)$ is reliable at time $t$; \\
  0, ~ otherwise.  \label{indicator_var}
  \end{tabular}
  \right.
\end{align}
Recall that $z_i$ and $w_i$ are the value and weight for $i\in \calV =\{1, \cdots, n\}$. For each $(j,i)\in \calE$, we define $z_{n_{ji}}$ and $w_{n_{ji}}$ as
\begin{align}
z_{n_{ji}}[t]~&\triangleq ~\sigma_j[t]-\rho_{ji}[t], ~\text{and}~  \label{z virtual}\\
w_{n_{ji}}[t]~&\triangleq ~\tilde{\sigma}_j[t]-\tilde{\rho}_{ji}[t], \label{w virtual}
\end{align}
with $z_{n_{ji}}[0]={\bf 0}\in \reals^d$ and $w_{n_{ji}}[0]=0\in \reals$.

Let $m=n+|\calE|$. 
We next show that the evolution of $z$ and $w$ can be described in a matrix form. 
 Since the update of value $z$ and weight $w$ are identical, for ease of exposition, henceforth, we focus on the value sequence $z$.

From steps 6 -- 10 of Algorithm \ref{alg:ps convergence rate}, we know
\begin{align}
\rho_{ji}[t]&=\sfB_{(j,i)}[t]\sigma^+_j[t]+(1-\sfB_{(j,i)}[t])\rho_{ji}[t-1],\label{e_rho_reform 1}. 
\end{align}
By \eqref{indicator_var}, \eqref{z virtual} and \eqref{e_rho_reform 1},  for each $i\in \calV$, the update of $z_i$ is
\begin{align}
\label{z[t]_reform 1}
\begin{cases}
z^+_{i}[t]&=\frac{z_{i}[t-1]}{d_i^o+1} + \sum_{j\in\calI_i}\sfB_{(j,i)}[t]\left(\frac{z_j[t-1]}{d_j^{o}+1}+z_{n_{ji}}[t-1]\right),\\
z_{i}[t] &= \frac{z^+_{i}[t]}{d_i^o+1}.
\end{cases}
\end{align}
Thus,
\begin{align}\label{matrix zi new}
z_i[t] &= \frac{z_{i}[t-1]}{\pth{d_i^o+1}^2}  + \sum_{j\in \calI_i} \frac{\sfB_{(j,i)}[t]}{\pth{d_i^o+1}\pth{d_j^o+1}} z_j[t-1] + \sum_{j\in \calI_i} \frac{\sfB_{(j,i)}[t]}{\pth{d_i^o+1}} z_{n_{ji}}[t-1].
\end{align}
%
Similarly, we get
\begin{align}\label{matrix z edge new}
\nonumber
z_{n_{ji}}[t]~ &= \frac{z_{j}[t-1]}{\pth{d_j^o+1}^2}  + \sum_{k\in \calI_j} \frac{\sfB_{(k,j)}[t]}{\pth{d_k^o+1}\pth{d_j^o+1}} z_k[t-1]  
+ \sum_{k\in \calI_j} \frac{\sfB_{(k,j)}[t]}{d_j^o+1} z_{n_{kj}}[t-1]\\
&\quad 
+ \frac{1- \sfB_{(j,i)}[t]}{d_j^o+1} z_j[t-1] + \pth{1-\sfB_{(j,i)}[t]} z_{n_{ji}}[t-1].
\end{align}
The detailed derivation of \eqref{matrix z edge new} can be found in Appendix \ref{app: proof of eqn}. 
Thus, we construct a matrix ${\bf M}[t]\in \reals^{m\times m}$ with the following structure:  
\begin{align}
\nonumber
&{\bf M}_{i,i}[t]\triangleq \frac{1}{\pth{d_i^{o}+1}^2}; \\
\nonumber
&{\bf M}_{j,i}[t]\triangleq \frac{\sfB_{(j,i)}[t]}{\pth{d_i^{o}+1}\pth{d_j^{o}+1}}, ~ \forall ~ j\in \calI_i; \\
\nonumber
&{\bf M}_{n_{ji},i}[t]\triangleq \frac{\sfB_{(j,i)}[t]}{d_i^{o}+1}, ~ \forall ~ j\in \calI_i;\\
\nonumber
&{\bf M}_{j, n_{ji}}[t] \triangleq \frac{1}{\pth{d_j^{o}+1}^2} + \frac{1-\sfB_{(j,i)}[t]}{d_j^{o}+1}; \\
\nonumber
&{\bf M}_{k, n_{ji}}[t] \triangleq \frac{\sfB_{(k,j)}[t]}{\pth{d_k^{o}+1}\pth{d_j^{o}+1}}, ~~ \forall ~ k\in \calI_j; \\
\nonumber
&{\bf M}_{n_{kj}, n_{ji}}[t] \triangleq \frac{\sfB_{(k,j)}[t]}{d_j^{o}+1},~~ \forall ~ k\in \calI_j; \\
&{\bf M}_{n_{ji}, n_{ji}}[t] \triangleq 1- \sfB_{(j,i)}[t].
\label{matrix con 2}
\end{align}
and any other entry in ${\bf M}[t]$ be zero.
It is easy to check that the obtained matrix ${\bf M}[t]$ is row stochastic. 
Let ${\bf \Psi}(r,t)$ be the product of $t-r+1$ row-stochastic matrices
\begin{align*}
{\bf \Psi}(r,t)&\triangleq \prod_{\tau=r}^t\, {\bf M}[\tau]={\bf M}[r] {\bf M}[r+1]\cdots {\bf M}[t],
\end{align*}
with $r\le t$. In addition, ${\bf \Psi}(t+1,t)\triangleq {\bf I}$ by convention. 

For ease of exposition, without loss of generality, let us fix a one-to-one mapping between $\{n+1, \cdots, m\}$ and $(j,i)\in \calE$. Thus, for each non-virtual agent $i\in \calV =\{1, \cdots, n\}$, we have
\begin{align}
z_i[t] 
& = \sum_{j=1}^{m} z_j[0] {\bf \Psi}_{ji}(1,t) = \sum_{j=1}^{n} y_j {\bf \Psi}_{ji}(1,t),
\label{evoz}
\end{align}
where the last equality holds due to $z_j[0] =y_j$ for $i\in \calV$ and $z_j[0]=0$ for $j\notin \calV$. 
%
%
%
%
%
Similar to \eqref{evoz}, for the weight evolution, for each $i\in \{1, \cdots, m\}$, we have
\begin{align}
w_i[t]
=\sum_{j=1}^n w_j[0] {\bf \Psi}_{ji}(1,t),
\label{evow}
\end{align}
%
Using ergodic coefficients and some celebrated results obtained by Hajnal \cite{Hajnal58}, we show the following thoerem. 
\begin{theorem}
\label{rps convergence rate}
Under Algorithm \ref{alg:ps convergence rate}, at each agent $i\in \calV =\{1, \cdots, n\}$,
\begin{align*}
\norm{\frac{z_i[t]}{w_i[t]} -\frac{1}{n}\sum_{k=1}^n y_k} \le  \frac{\sum_{k=1}^{n}y_k}{n\beta^{nB+1}}\gamma^{\lfloor\frac{t}{nB+1}\rfloor},
\end{align*}
where $\beta \triangleq \frac{1}{\max_{i\in \calV} (d_i^{o}+1)^2}$ and $\gamma \triangleq 1-\beta^{nB+1}$
\end{theorem}
Here we use $\norm{\cdot}$ to denote $\ell_2$ norm. The proof of Theorem \ref{rps convergence rate} can be found in Appendix \ref{app: proof of thm1}. 
Note that the above convergence rate might not be tight. Better rates might be obtained by carefully exploring the structures of the communication graphs.

\section{Robust Distributed Dual Averaging Method}
We apply Algorithm \ref{alg:ps convergence rate} to distributed dual averaging method as information fusion primitive. 
Throughout this section, we assume that each agent $i$ knows a private cost function $h_i: \calX \to \reals$, where \\
(A) ~~ $\calX\subseteq \reals^d$ is nonempty, convex and compact; and \\
(B) ~~  $h_i$ is convex and $L$--Lipschitz continuous with respect to $\ell_2$ norm, i.e., for all $x,y\in \calX$,
\begin{align}
\label{Lip}
\norm{h_i(x)-h_i(y)}\le L\norm{x-y}, \forall i\in \calV
\end{align}
We are interested in solving 
\begin{align}
\label{goallk}
\min_{x\in \calX}~~h(x)&\triangleq \,\frac{1}{n}\sum_{i=1}^{n}h_i(x).
\end{align}
using a multi-agent network where the communication links may suffer packet-dropping failures. 
Let $X^*$ be the collection of optimal solutions of $h$ subject to $\calX$.
Since $\calX\subseteq \reals^d$ is a nonempty, convex and compact, $X^*$ is also nonempty, convex and compact.

In addition to the estimate sequence $\{x[t]\}_{t=0}^{\infty}$, in dual averaging method, there is an additional sequence $\{z[t]\}_{t=0}^{\infty}$ in the dual space that essentially aggregates all the sub-gradients generated so far. In addition, the dual averaging scheme involves a \emph{proximal function} $\psi: \reals^d\to \reals$ that is strongly convex. In this paper, we choose $\psi$ to be $1$--strongly convex with respect to $\ell_2$ norm, that is
\begin{align*}
\psi(y)\ge \psi(x)+\iprod{\nabla \psi(x)}{~ y-x}+\frac{1}{2}\norm{x-y}^2,
\end{align*}
for $x, y\in \reals^d$. In addition, we assume that $\psi\ge 0$ and $\argmin_{x}\psi(x)={\bf 0}\in \reals^d$, which is also referred as proximal center.
This choice of $\psi$ is rather standard \cite{Duchi2012,Tsianos2012CDC}. This proximal function, in a sense, is used to ``smooth" the update of the primal sequence $\{x[t]\}_{t=0}^{\infty}$.

One typical iterate sequence under dual averaging method is as follows. Initializing $z[0]=x[0]={\bf 0}\in \reals^d$, for iteration ($t\ge 0$), compute $g[t]\in \partial h(x[t])$, and update $z$ and $x$ as
\begin{align}
z[t+1]~&=~z[t]+g[t],  \label{update dual}\\
x[t+1]~&=~\prod\nolimits_{x\in \reals^d}^{\psi} \pth{z[t+1], \alpha[t]}, \label{update primal}
\end{align}
where $\prod\nolimits_{x\in \reals^d}^{\psi}(\cdot)$ is the projection operator defined as
\begin{align}
\prod\nolimits_{x\in \reals^d}^{\psi} \pth{z,\alpha}~\triangleq ~\argmin_{x\in \reals^d}\sth{\iprod{z}{x}+\frac{1}{\alpha}\psi(x)}.
\end{align}
From (\ref{update primal}), we know that the update of $x$ is based on all the subgradients generated so far, and all these subgradients are weighted equally. 
%
%
%
The convergence rate of the dual averaging method is $O(\frac{1}{\sqrt{t}})$, which is faster than the subgradient method whose convergence rate is $O(\frac{\log t}{\sqrt{t}})$. Besides, the constants of the dual averaging method are often smaller \cite{nesterov2009primal}.

%
%
%

Next we present our {\em Robust Push-Sum Distributed Dual Averaging} (RPSDA) method. In our RPSDA, each agent $i$ locally keeps
\begin{itemize}
\item  estimate sequence $\{x_i[t]\}_{t=0}^{\infty}$,
\item  gradient aggregation (value) sequence $\{z_i[t]\}_{t=0}^{\infty}$, and
\item  weight sequence $\{w_i[t]\}_{t=0}^{\infty}$,
\end{itemize}
where $x_i[0]=z_i[0]={\bf 0}\in \reals^d$ and $w_i[0]=1\in \reals$.
In addition, let $\{\alpha[t]\}_{t=0}^{\infty}$ be a sequence of positive stepsizes. We will specify the choice of $\alpha[t]$ in our statement of theorem.
\begin{algorithm}
\caption{RPSDA}
\label{alg:ps psdda}
{\em Initialization}:  $z_i[0]=x_{i}[0]=\sigma_i[0]={\bf 0}\in \reals^d$, $\tilde{\sigma}_i[0]=0\in \reals$, $w_i[0]=1\in \reals$, $\rho_{ji}[0]={\bf 0}\in \reals^d$ and $\tilde{\rho}_{ji}[0]=0\in \reals$ for each incoming link, i.e., $j\in \calI_i$.

\For{$t\ge 1$}
{
$\sigma^{+}_i[t]  \gets  \sigma_i[t-1] + \frac{z_i[t-1]}{d_i^{o}+1}$,
$\tilde{\sigma}^{+}_i[t]  \gets  \tilde{\sigma}_i[t-1] + \frac{w_i[t-1]}{d_i^{o}+1}$\;

Broadcast $\pth{\sigma^{+}_i[t], \tilde{\sigma}^{+}_i[t]}$ to outgoing neighbors\;

\For {each incoming link $(j,i)$}
{\eIf{message $\pth{\sigma^{+}_j[t], \tilde{\sigma}^{+}_j[t]}$ is received}
{$\rho_{ji}[t] \gets \sigma^{+}_j[t]$, ~~ $\tilde{\rho}_{ji}[t] \gets \tilde{\sigma}^{+}_j[t]$\;}
{ $\rho_{ji}[t] \gets \rho_{ji}[t-1]$, ~~$\tilde{\rho}_{ji}[t] \gets \tilde{\rho}_{ji}[t-1]$\;}
$ z_i^{+}[t] \gets \frac{z_i[t-1]}{d_i^o+1} + \sum_{j\in \calI_i} \pth{\rho_{ji}[t] - \rho_{ji}[t-1]}$,  $w_i^{+}[t]  \gets \frac{w_i[t-1]}{d_i^o+1} + 
\sum_{j\in \calI_i} \pth{\tilde{\rho}_{ji}[t] -\tilde{\rho}_{ji}[t-1]}$.
}

$\sigma_i[t]  \gets  \sigma^{+}_i[t] + \frac{z_i^{+}[t]}{d_i^{o}+1}$,
$\tilde{\sigma}_i[t]  \gets  \tilde{\sigma}^{+}_i[t] + \frac{w_i^{+}[t]}{d_i^{o}+1}$,
$z_i[t] \gets \frac{z_i^+[t]}{d_i^{o}+1}$,
$w_i[t] \gets \frac{w_i^+[t]}{d_i^{o}+1}$.

Compute a subgradient $g_i[t-1]\in \partial h_i\pth{x_i[t-1]}$\; 

$z_i[t] \gets z_i[t] +g_i[t-1]$\; 

$x_i[t]\gets
\prod\nolimits_{\calX}^{\psi}\pth{\frac{z_i[t]}{w_i[t]}, \alpha[t-1]}$\;
}
\end{algorithm}
Note that the only difference between Algorithm \ref{alg:ps psdda} and Algorithm \ref{alg:ps convergence rate} is that 
a subgradient is computed and added to the local value $z$. One importantly, the local estimate $x$ is updated using dual averaging update.

For ease of exposition, let $g_i[r]=0$  for each virtual agent $ i\in \{n+1, \cdots, m\}$ and $r\ge 0$. 
Similar to \eqref{evoz} and \eqref{evow}, we have 
\begin{align*}
 z_i[t] & =\sum_{r=0}^{t-1} \sum_{j=1}^n g_j[r] {\bf \Psi}_{j,i}(r, t)\\
 w_i[t] & =\sum_{j=1}^n {\bf \Psi}_{j,i}(1, t). 
\end{align*}

Let $\bar{z}[t]\triangleq \frac{1}{n}\sum_{i=1}^n z_i[t]$.
We have
\begin{align}
\bar{z}[t]=\frac{1}{n}\sum_{i=1}^m z_i[t]
=\frac{1}{n}\sum_{r=0}^{t-1}\sum_{i=1}^{n}g_i[r]. 
\label{limitingZ}
\end{align}
Let $\{\alpha[t]\}_{t=0}^{\infty}$ be a sequence of non-increasing stepsizes. 
%
For each agent $i\in \calV$, we define the running average of $x_i[t]$, denoted by $\hat{x}_i[T]$, as follows:
$$\hat{x}_i[T]=\frac{1}{T}\sum_{t=1}^{T}x_i[t].$$ 

\begin{theorem}
\label{main}
Let $x^*\in X^*$, and suppose that $\psi(x^*)\le R^2$.
Let $\{\alpha[t]=\frac{A}{\sqrt{t}}\}_{t=1}^{\infty}$ with $\alpha[0]=A$ be the sequence of stepsizes used in Algorithm \ref{alg:ps psdda} for some positive constant $A$. Then, for $T\ge nB+1$, we have for all $j\in \calV$,
\begin{align*}
h\pth{\hat{x}_j[T]}-h(x^*)&\le \frac{2L^2A}{T} (2\sqrt{T} +1) + \frac{R^2}{A\sqrt{T}}\\
&\quad + \frac{3L^2A}{\beta^{nB+1}(1-\gamma^{\frac{1}{nB+1}})  \gamma^{\frac{nB}{nB+1}}} \frac{2\sqrt{T} +1}{T}.
\end{align*}
\end{theorem}
Recall from Theorem \ref{rps convergence rate} that $\beta \triangleq \min_{i\in \calV} \frac{1}{(d_i^o+1)^2}$ and $\gamma \triangleq 1 - \beta^{nB+1}$. 
Similar to the results in Theorem \ref{rps convergence rate}, the rate in Theorem \ref{main} might be improved by carefully exploring the structures of the communication graphs. 

Note that Theorem \ref{main} holds for any positive constant $A$. Optimizing over $A$, the constant hidden in $O(\frac{1}{\sqrt{T}})$ can be improved. 


%
%
%
%
%
\bibliographystyle{abbrv}
\bibliography{PSDA_DL}

\appendix 

\section{Proof of Equation \eqref{matrix z edge new}} 
\label{app: proof of eqn}
By \eqref{z virtual}, we have 
\begin{align*}
z_{n_{ji}}[t]~ &= \sigma_j[t]-\rho_{ji}[t] \\
& = \sigma_j^+[t] + \frac{z_j^+[t]}{d_j^o +1} - \pth{\sfB_{(j,i)}[t]\sigma^+_j[t]+(1-\sfB_{(j,i)}[t])\rho_{ji}[t-1]}\\ 
& = (1-\sfB_{(j,i)}[t]) \pth{\sigma_j[t-1] + \frac{z_j[t-1]}{d_j^o +1}}  - (1-\sfB_{(j,i)}[t])\rho_{ji}[t-1] + \frac{z_j^+[t]}{d_j^o +1}\\
& = (1-\sfB_{(j,i)}[t]) z_{n_{ji}}[t-1] + (1-\sfB_{(j,i)}[t]) \frac{z_j[t-1]}{d_j^o +1}  + \frac{z_j^+[t]}{d_j^o +1}\\
&=\pth{\frac{1}{\pth{d_j^o+1}^2} + \frac{1- \sfB_{(j,i)}[t]}{d_j^o+1}} z_j[t-1]  + \sum_{k\in \calI_j} \frac{\sfB_{(k,j)}[t]}{\pth{d_k^o+1}\pth{d_j^o+1}} z_k[t-1] \\ 
&\quad + \sum_{k\in \calI_j} \frac{\sfB_{(k,j)}[t]}{d_j^o+1} z_{n_{kj}}[t-1] + \pth{1-\sfB_{(j,i)}[t]} z_{n_{ji}}[t-1].
\end{align*}

\section{Proof of Theorem \ref{rps convergence rate}}
\label{app: proof of thm1}
In this subsection, we investigate the convergence behavior of ${\bf \Psi}(r,t)$ (where $r\le t$) using ergodic coefficients and some celebrated results obtained by Hajnal \cite{Hajnal58}.

Given a row stochastic matrix ${\bf A}$,
 coefficients of  ergodicity   $\delta({\bf A})$ and $\lambda({\bf A})$ are defined as:
\begin{align}
\delta({\bf A}) & \triangleq   \max_j ~ \max_{i_1,i_2}~\left | {\bf A}_{i_1 j}-{\bf A}_{i_2 j}\right |, \label{e_delta} \\
\lambda({\bf A}) & \triangleq   1 - \min_{i_1,i_2} \sum_j \min\{{\bf A}_{i_1 j}, {\bf A}_{i_2 j}\}. \label{e_lambda}
\end{align}
Informally speaking, the coefficients of ergodicity defined in \eqref{e_delta} and \eqref{e_lambda} characterize the ``difference" between any pair of rows of the given row-stochastic matrix ${\bf A}$. It is easy to see that  $0\leq \delta({\bf A}) \leq 1$, $0\leq \lambda({\bf A}) \leq 1$, and that the rows of ${\bf A}$ are identical if and only if $\delta({\bf A})=0=\lambda({\bf A})$. In addition, the ergodic coefficients $\delta(\cdot)$ and $\lambda(\cdot)$ have the following connection.
%

\begin{proposition}\cite{Hajnal58}
\label{claim_delta}
For any $p$ square row stochastic matrices ${\bf Q}[1],{\bf Q}[2],\dots {\bf Q}[p]$, it holds that
\begin{align}
\delta({\bf Q}[1]{\bf Q}[2]\ldots {\bf Q}[p]) ~\leq ~
 \Pi_{k=1}^p ~ \lambda({\bf Q}[k]).
\end{align}
\end{proposition}

Proposition \ref{claim_delta} implies that if $\lambda({\bf Q}[k])\leq 1-c$ for some $c>0$ and for all $1\le k\le p$, then $\delta({\bf Q}[1],{\bf Q}[2]\cdots {\bf Q}[p])$ goes to zero exponentially fast as $p$ increases.
Next we show that, for sufficiently large $t$, it holds that $\lambda({\bf \Psi}(1,t))\le 1-\beta^{nB}$, where $\beta \triangleq \frac{1}{\max_{i\in \calV} (d_i^{o}+1)^2}$. To prove this claim, we need the following lemma, whose proof is rather standard and is omitted. 
\begin{lemma}
Suppose that $t-r+1\ge nB+1$ and $B\ge 1$. Then every entry in ${\bf \Psi}(r, t)$ is lower bounded by $\beta^{nB+1}$. 
\label{c1}
\end{lemma}


%
By Proposition \ref{claim_delta} and Lemma \ref{c1}, we are able to show Lemma \ref{c2}, which says that the difference between any pair of rows in ${\bf \Psi}(r, t)$ goes to 0 exponentially fast. 
\begin{lemma}
\label{c2}
For $r\le t$, it holds that $\delta\pth{{\bf \Psi}(r, t)}\le \gamma^{\lfloor\frac{t-r+1}{nB+1}\rfloor},$
where $\gamma=1-\beta^{nB+1}$.
\end{lemma}
The proof of Lemma \ref{matrix con 2} is a straightforward application of Proposition \ref{claim_delta} and Lemma \ref{c1}; thus is omitted.

\begin{theorem}
\label{rps convergence rate 2}
Under Algorithm \ref{alg:ps convergence rate}, at each agent $i\in \calV =\{1, \cdots, n\}$,
\begin{align*}
\norm{\frac{z_i[t]}{w_i[t]} -\frac{1}{n}\sum_{k=1}^n y_k} \le  \frac{\sum_{k=1}^{n}y_k}{n\beta^{nB+1}}\gamma^{\lfloor\frac{t}{nB+1}\rfloor},
\end{align*}
where $\norm{\cdot}$ is the $\ell_2$ norm. 
\end{theorem}
\begin{proof}
\begin{align*}
\norm{\frac{z_i[t]}{w_i[t]} -\frac{1}{n}\sum_{k=1}^n y_k}  &= \norm{\frac{\sum_{j=1}^{n}y_j{\bf \Psi}_{j,i}(1,t)}{\sum_{j=1}^{n}{\bf \Psi}_{j,i}(1,t)} -\frac{1}{n}\sum_{k=1}^n y_k}\\
& = \norm{\frac{n\sum_{j=1}^{n}y_j{\bf \Psi}_{j,i}(1,t) - \sum_{k=1}^n y_k \sum_{j=1}^{n}{\bf \Psi}_{j,i}(1,t)}{n \sum_{j=1}^{n}{\bf \Psi}_{j,i}(1,t)}} \\
&= \frac{\norm{\sum_{j=1}^{n}y_j\sum_{k=1}^{n}\pth{{\bf \Psi}_{j,i}(1,t) -{\bf \Psi}_{k,i}(1,t)}}}{n\sum_{j=1}^{n}{\bf \Psi}_{j,i}(1,t)}\\
&\le \frac{\sum_{j=1}^{n}y_j n\gamma^{\lfloor\frac{t}{nB+1}\rfloor}}{n\sum_{j=1}^{n}{\bf \Psi}_{j,i}(1,t)}, ~~~~ \text{by Lemma \ref{c2}}\\
&\le \frac{\sum_{k=1}^{n}y_k}{n\beta^{nB+1}}\gamma^{\lfloor\frac{t}{nB+1}\rfloor}, ~~~ \text{by Lemma \ref{c1}},
\end{align*}
and the proof is complete. 
\end{proof}

\section{Proof of Theorem \ref{main}}
\label{app: thm main}
The proof of Theorem \ref{main} relies on a couple of auxiliary lemmas, stated and proved next. 
We need the sequence $\{y(t)\}_{t=1}^{\infty}$ that is defined by the projection of $\bar{z}[t]$:
\begin{align}
\label{centralized}
y[t]\triangleq \prod\nolimits_{\calX}^{\psi}\pth{\bar{z}[t],\alpha[t-1]}.
\end{align}

 Using the standard convexity arguments as in \cite{Tsianos2012CDC}, the following lemma holds. 
 Note that the summation on the RHS is over all agents in the {\em original graph} $G(\calV, \calE)$ rather than the augmented graph $G^a(\calV^a, \calE^a)$.

\begin{lemma}
\label{convergence}
For any $x^*\in \calX$, it holds that
\begin{align*}
 h(\hat{x}_j[T])-h(x^*) & \le  \frac{L^2}{T}\sum_{t=1}^T \alpha[t-1]+\frac{1}{T\alpha[T]}\psi(x^*)+\frac{2L}{nT}\sum_{t=1}^{T}\sum_{i=1}^{n}\alpha[t-1]\norm{\bar{z}[t] -\frac{z_i[t]}{w_i[t]}} \\
&\quad +\frac{L}{T}\sum_{t=1}^{T}\alpha[t-1]\norm{ \bar{z}[t] -\frac{z_j[t]}{w_j[t]}}.
\end{align*}
\end{lemma}
\begin{proof}
Adding and subtracting $h\pth{\hat{y}[T]}$
\begin{align*}
h\pth{\hat{x}_j[T]}-h(x^*)&=h\pth{\hat{y}[T]}-h(x^*)+h\pth{\hat{x}_j[T]}-h\pth{\hat{y}[T]}\\
&\le h(\hat{y}[T])-h(x^*)+L \norm{ \hat{x}_j[T]-\hat{y}[T]}\\
&\le \frac{1}{T}\sum_{t=1}^{T}\pth{h(y[t])-h(x^*)}+\frac{L}{T}\sum_{t=1}^{T}\norm{ x_j[t]-y[t]}.
\end{align*}
The first inequality holds from $L$--Lipschitz contunity; and the second inequality is true due to the convexity of $h$ as well as the definition of the running averages $\hat{x}_j[T]$ and $\hat{y}[T]$. Now we add and subtract $\sum_{t=1}^{T}\frac{1}{n}\sum_{i=1}^n h_i\pth{x_i[t]}$ and use convexity and $L$--Lipschitz continuity of the component functions $h_i(x)$ to get
\begin{align}
\label{gap 1}
\nonumber
h\pth{\hat{x}_j[T]}-h(x^*)&\le
\frac{1}{T}\sum_{t=1}^{T}\frac{1}{n}\sum_{i=1}^{n}\pth{h_i(y[t])-h_i(x_i[t])}+\frac{1}{T}\sum_{t=1}^{T}\frac{1}{n}\sum_{i=1}^{n}\pth{h_i(x_i[t])-h_i(x^*)}\\
\nonumber
&+\frac{L}{T}\sum_{t=1}^{T}\norm{x_j[t]-y[t]}\\
\nonumber
&\le \frac{1}{T}\sum_{t=1}^{T}\frac{1}{n}\sum_{i=1}^{n}L\norm{ x_i[t]-y[t]}+\frac{1}{T}\sum_{t=1}^{T}\frac{1}{n}\sum_{i=1}^{n}\iprod{g_i[t]}{x_i[t]-x^*}
 +\frac{L}{T}\sum_{t=1}^{T}\norm{x_j[t]-y[t]}\\
\nonumber
&\le \frac{L}{Tn}\sum_{t=1}^{T}\alpha[t-1]\sum_{i=1}^{n}\norm{ \frac{z_i[t]}{w_i[t]}-\bar{z}[t]}+\frac{1}{nT}\sum_{t=1}^{T}\sum_{i=1}^{n}\iprod{g_i[t]}{x_i[t]-x^*}\\
&+\frac{L}{T}\sum_{t=1}^{T}\alpha[t-1]\norm{\frac{z_j[t]}{w_j[t]}-\bar{z}[t]},
\end{align}
%
For the second term in (\ref{gap 1}), we have
\begin{align*}
\sum_{i=1}^{n}\iprod{g_i[t]}{x_i[t]-x^*}&=\sum_{i=1}^{n}\iprod{g_i[t]}{y[t]-x^*} +\sum_{i=1}^{n}\iprod{g_i[t]}{x_i[t]-y[t]}\\
&=\iprod{\sum_{i=1}^{n} g_i[t]}{y[t]-x^*}+\sum_{i=1}^{n}\iprod{g_i[t]}{x_i[t]-y[t]}.
\end{align*}
Let $g[t]=\frac{1}{n}\sum_{i=1}^n g_i[t]$. It holds that
\begin{align}
\label{imme 1}
\bar{z}[t]&=\frac{1}{n}\sum_{r=0}^{t-1}\sum_{i=1}^n g_i[r], 
\end{align}
and that 
\begin{align*}
y[t]=\prod\nolimits_{\calX}^{\psi}\pth{\bar{z}[t],\alpha[t-1]}=\prod\nolimits_{\calX}^{\psi}\pth{\sum_{\tau=1}^tg[t],\alpha[t-1]}. 
\end{align*}
Thus,
\begin{align}
\label{gap 2}
&\sum_{t=1}^T \frac{1}{n}\iprod{ \sum_{i=1}^{n}g_i[t]}{ y[t]-x^*}
=\sum_{t=1}^{T}\iprod{ g[t]}{y[t]-x^*}
=\frac{L^2}{2}\sum_{t=1}^T \alpha[t-1]+\frac{1}{\alpha[T]}\psi(x^*),
\end{align}
where the last inequality holds since $\norm{g[r]}\le L$ for all $r\ge 0$. 
In addition, 
\begin{align}
\label{gap 3}
\sum_{i=1}^{n}\iprod{g_i[t]}{x_i[t]-y[t]}
\le L\sum_{i=1}^{n} \alpha[t-1]\norm{\frac{z_i[t]}{w_i[t]}- \bar{z}[t]}. 
\end{align}

Plugging (\ref{gap 2}) and (\ref{gap 3}) back to (\ref{gap 1}), we get
\begin{align*}
h(\hat{x}_j[T])-h(x^*)& \le  \frac{L^2}{T}\sum_{t=1}^T \alpha[T-1]+\frac{1}{T\alpha[T]}\psi(x^*) +\frac{2L}{nT}\sum_{t=1}^{T}\sum_{i=1}^{n}\alpha[t-1]\norm{\bar{z}[t] -\frac{z_i[t]}{w_i[t]}} \\
&\quad +\frac{L}{T}\sum_{t=1}^{T}\alpha[t-1]\norm{ \bar{z}[t] -\frac{z_j[t]}{w_j[t]}},
\end{align*}
proving the proposition.

\end{proof}

\vskip \baselineskip

To complete the convergence analysis,  we  need to bound each term $\norm{\bar{z}[t]-\frac{z_i[t]}{w_i[t]}}$ for any agent $i$ and any iteration $t\ge 1$. Our analysis is different from that in \cite{Tsianos2012CDC}, due to ${\bf M}[t]$'s dependency on time $t$.

\begin{lemma}
\label{mixing error}
When $t\ge nB+1$, for each $i\in \calV$, it holds that
\begin{align*}
\norm{\bar{z}[t]-\frac{z_i[t]}{w_i[t]}}\le \frac{L}{\beta^{nB+1}(1-\gamma^{\frac{1}{nB+1}})  \gamma^{\frac{nB}{nB+1}}}. 
\end{align*}
\end{lemma}

\begin{proof}
Similar to the proof of Theorem \ref{rps convergence rate}, it can be shown that 
\begin{align*}
\norm{\bar{z}[t]-\frac{z_i[t]}{w_i[t]}}&=
\norm{\frac{1}{n} \sum_{r=0}^{t-1}\sum_{j=1}^n g_j[r] -\frac{\sum_{r=0}^{t-1}\sum_{j=1}^n g_j[r] {\bf \Psi}_{j,i}(r,t)}{\sum_{j=1}^{n}{\bf \Psi}_{j,i}(1,t)} }\\
&= \norm{\frac{\sum_{r=0}^{t-1} \sum_{j=1}^n g_j[r] \sum_{k=1}^n \pth{ {\bf \Psi}_{k,i}(1, t) -  {\bf \Psi}_{j,i}(r, t)}   }{n \sum_{j=1}^{n}{\bf \Psi}_{j,i}(1,t)} }\\
&\le \frac{\norm{\sum_{r=0}^{t-1} \sum_{j=1}^n g_j[r] \sum_{k=1}^n \pth{ {\bf \Psi}_{k,i}(1, t) -  {\bf \Psi}_{j,i}(r, t)}  } }{ n n  \beta^{nB+1}} \\
&\le \frac{ L \sum_{r=0}^{t-1} \sum_{j=1}^n \sum_{k=1}^n \norm{ {\bf \Psi}_{k,i}(1, t) -  {\bf \Psi}_{j,i}(r, t)}}{n^2  \beta^{nB+1}}
\end{align*}
We know that 
\begin{align*}
\norm{ {\bf \Psi}_{k,i}(1, t) -  {\bf \Psi}_{j,i}(r, t)} & = \norm{\sum_{p=1}^m {\bf \Psi}_{k,p}(1, r-1) {\bf \Psi}_{p,i}(r, t) -  {\bf \Psi}_{j,i}(r, t)}\\
&\le \sum_{p=1}^m {\bf \Psi}_{k,p}(1, r-1) \norm{{\bf \Psi}_{p,i}(r, t) -  {\bf \Psi}_{j,i}(r, t)}\\
&\le \gamma^{\lfloor \frac{t-r+1}{nB+1}\rfloor}. 
\end{align*}
Thus, we have 
\begin{align*}
\norm{\bar{z}[t]-\frac{z_i[t]}{w_i[t]}}&\le \frac{L}{\beta^{nB+1}(1-\gamma^{\frac{1}{nB+1}})  \gamma^{\frac{nB}{nB+1}}}.
\end{align*}
\end{proof}

Now we are ready to finish the proof of Theorem \ref{main}. 

\begin{proof}[\bf Proof of Theorem \ref{main}]
By the assumption that $\psi(x^*)\le R^2$ and Lemmas \ref{convergence} and \ref{mixing error}, we have
\begin{align}
\label{aaa3}
h(\hat{x}_j[T])-h(x^*)\le  \frac{L^2}{T}\sum_{t=1}^T \alpha[t-1]+\frac{1}{T\alpha[T]}R^2+\frac{3L}{T}\sum_{t=1}^{T}\alpha[t-1] \frac{L}{\beta^{nB+1}(1-\gamma^{\frac{1}{nB+1}})  \gamma^{\frac{nB}{nB+1}}}.
\end{align}
For the chosen step-sizes  $\alpha[t]=\frac{A}{\sqrt{t}}$ for $t\ge 1$ and $\alpha[0]=A$, we have \begin{align}
\label{stepsize ub}
\sum_{t=1}^T \alpha[t-1] = \sum_{t=1}^{T-1} \frac{A}{\sqrt{t}} +A \le 2\sqrt{T} A +A. 
\end{align}
Plugging the above upper bound on the step-sizes (\ref{stepsize ub}) back to \eqref{aaa3}, the bound in the statement of Theorem \ref{main} is obtained. 

\end{proof}

\end{document}